\documentclass[pra,preprint,superscriptaddress,amssymb]{revtex4}
\usepackage{graphicx}
\usepackage{bm}
\usepackage{amsmath,amssymb}
\usepackage{color}
\usepackage{amsthm} 

\newtheorem{definition}{Definition}
\newtheorem{lemma}{Lemma}
\newtheorem{theorem}{Theorem}
\newtheorem{proposition}{Proposition}

\begin{document}
\title{
Merlinization of complexity classes above BQP
} 
\author{Tomoyuki Morimae}
\email{morimae@gunma-u.ac.jp}
\affiliation{Department of Computer Science,
Gunma University, 1-5-1 Tenjincho Kiryushi
Gunma, 376-0052, Japan}
\author{Harumichi Nishimura}
\email{
hnishimura@is.nagoya-u.ac.jp
}
\affiliation{
Graduate School of Informatics,
Nagoya University,
Furocho, Chikusaku, Nagoya, Aichi,
464-8601, Japan
}

\begin{abstract}
We study how complexity classes above BQP,
such as postBQP, ${\rm postBQP}_{\rm FP}$, and SBQP,
change if we ``Merlinize" them, i.e., 
if we allow an extra input quantum state (or classical bit string) 
given by Merlin as witness.
Main results are the following three:
First, the Merlinized version of postBQP is equal to 
PSPACE.
Second, if the Merlinized postBQP is restricted
in such a way that the postselection probability is equal
to all witness states, then the class
is equal to PP.
Finally, the Merlinization does not change the class
SBQP.
\end{abstract}
\maketitle

\section{Introduction}
QMA (Quantum Merlin-Arthur) is a quantum version of NP (more precisely,
MA) first studied by 
Knill~\cite{Knill},
Kitaev~\cite{Kitaev}, 
and
Watrous~\cite{Watrous}. 

\begin{definition}
A language $L$ is in {\rm QMA} 
if and only if there exist polynomials
$w$, $m$, and a uniform family $\{Q_x\}_x$ 
of polynomial-size quantum circuits,
where $x$ is an instance with $|x|=n$, 
Arthur's circuit
$Q_x$ takes as input a $w(n)$-qubit quantum state (so called the witness)
sent from Merlin, 
and $m(n)$ ancilla qubits initialized in $|0\rangle$,
such that
\begin{itemize}
\item
if $x\in L$, then there exists a $w(n)$-qubit 
quantum state $\psi$ such that 
\begin{eqnarray*}
P_{Q_x(\psi)}(o=1)\ge a,
\end{eqnarray*}
\item
if $x\notin L$, then for any $w(n)$-qubit quantum state $\xi$,
\begin{eqnarray*}
P_{Q_x(\xi)}(o=1)\le b.
\end{eqnarray*}
\end{itemize}
Here, 
\begin{eqnarray*}
P_{Q_x(\xi)}(o=1)\equiv{\rm Tr}\Big[
(|1\rangle\langle1|\otimes I^{\otimes w(n)+m(n)-1})
Q_x(\xi\otimes |0\rangle\langle0|^{\otimes m(n)})Q_x^\dagger
\Big]
\end{eqnarray*}
is the probability that the circuit $Q_x$
on input $\xi\otimes |0\rangle\langle 0|^{\otimes m(n)}$
outputs $o=1$,
and 
$a-b\ge 1/poly(n)$.
Note that, without loss of generality, we can assume that the
yes witness $\psi$ is a pure state.
\end{definition}

QMA has a variant, which is called QCMA~\cite{AN}, where
the witness quantum state is replaced with a poly-length
classical bit string.
Whether ${\rm QMA}\neq{\rm QCMA}$ (i.e., quantum witnesses
are more powerful than classical ones) is one of long-standing
open problems.

\begin{definition}
The class {\rm QCMA} is defined similarly to 
{\rm QMA} except that the witness is
not a $w(n)$-qubit quantum state but a classical $w(n)$-bit string
(or, equivalently, a $w(n)$-qubit state in the computational basis).
\end{definition}

Due to the witness states given by the powerful Merlin,
QCMA and QMA are stronger than BQP. For example, it
is known that QMA can solve the group non-membership problem~\cite{Watrous},
but it is not known how the witness state for the group
non-membership problem can be generated
in quantum polynomial time.

Studying complexity classes above BQP has recently been attracting much
attentions because of several reasons. First,
studying these classes can give insights to understanding
why quantum theory has such a mathematical structure.
In particular, the existence of the so-called 
Popescu-Rohrlich box~\cite{PRbox} suggests that
quantum physics is not uniquely derived only from
the no-signaling principle. Physicists have therefore been interested
in reasons why quantum theory is as it is.
Several ``super quantum" computing models have been demonstrated
to have much stronger power than the standard polynomial-time
quantum computing~\cite{AbramsLloyd,Aaronson,AaronsonBoulandFitzsimons,
MN,Ciaran1,Ciaran2}. 
These results explain why quantum theory
should have the current form.
Second, studying complexity classes above BQP is related to
studying quantum supremacy of sub-universal quantum computing models.
It has been shown that several sub-universal quantum computing
models, such as IQP~\cite{IQP1,IQP2}, non-interacting bosons~\cite{AA},
and the DQC1 model~\cite{MFF,FKMNTT}, cannot be classically efficiently simulated
unless the polynomial hierarchy collapses.
To show the no-go results, some complexity classes 
above BQP, such as ${\rm postBQP}={\rm PP}$ and SBQP, are used.

Do the witness states by Merlin also
enhance such above-BQP classes?
In this paper, we study
how complexity classes above BQP
change if we ``Merlinize" them like QMA and QCMA.
In other words, we  
allow an extra input quantum state (or classical bit string) 
given by Merlin as witness.
More precisely, we define new classes,
QMA$_{\rm postBQP}$,
QCMA$_{\rm postBQP}$,
QMA$_{\rm postBQP}^*$,
QCMA$_{\rm postBQP}^*$,
QMA$_{\rm postBQP_{\rm FP}}$,
QCMA$_{\rm postBQP_{\rm FP}}$,
QMA$_{\rm SBQP}$,
and
QCMA$_{\rm SBQP}$,
as Merlinized versions of postBQP, postBQP$_{\rm FP}$,
and SBQP, respectively.
Definitions of postBQP, postBQP$_{\rm FP}$,
and SBQP are given in Sec.~\ref{known_classes}.
Our new eight classes are defined in Sec.~\ref{new_classes}.
Our results are summarized in Fig.~\ref{fig}.

In particular, our main results are the following three:
First, the quantum-witness Merlinized version of postBQP, 
${\rm QMA}_{\rm postBQP}$, is equal to 
PSPACE.
As we can easily see that the classical-witness Merlinized version
of postBQP, ${\rm QCMA}_{\rm postBQP}$ is equal to ${\rm NP}^{\rm PP}$,
this result implies that quantum witnesses are more powerful
than classical ones when Arthur has the strong power of 
postselections (unless ${\rm PSPACE}={\rm NP}^{\rm PP}$).
Second, if the Merlinized postBQP is restricted
in such a way that the postselection probability is equal
to all witness states, such classes, ${\rm Q(C)MA}_{\rm postBQP}^*$,
are equal to PP.
Finally, the Merlinization does not change the class
SBQP: ${\rm Q(C)MA}_{\rm SBQP}={\rm SBQP}$.
The second and third results mean that if Arthur has 
super-quantum powers that are weaker than postBQP, the computational
power of quantum witnesses is equivalent to that of classical ones.

\begin{figure}[htbp]
\begin{center}
\includegraphics[width=0.4\textwidth]{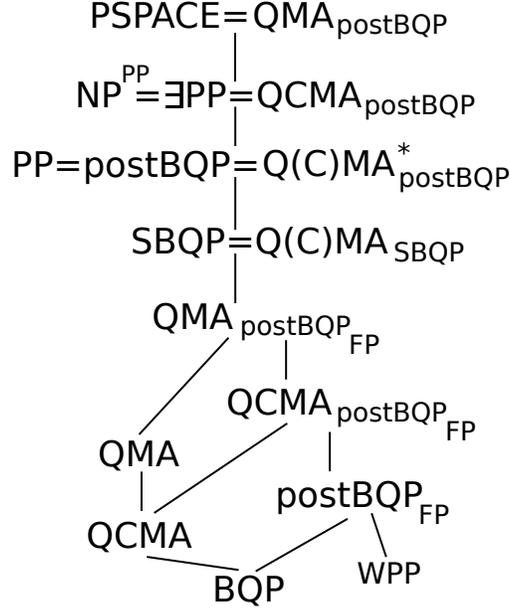}
\end{center}
\caption{Summary of results.}
\label{fig}
\end{figure}

\section{Known classes}
\label{known_classes}
In this section, we review definitions of known complexity classes.

First, 
the class postBQP was defined by Aaronson~\cite{Aaronson}, and it was shown
to be equal to PP.

\begin{definition}
A language $L$ is in {\rm postBQP} if and only if
there exist a polynomial $s$ and a uniform family
$\{Q_x\}_x$ of polynomial-size
quantum circuits, where $x$ is an instance with $|x|=n$,
such that
\begin{eqnarray*}
P_{Q_x}(p=1)\ge\frac{1}{2^{s(n)}}
\end{eqnarray*}
and
\begin{itemize}
\item
if $x\in L$, then 
$P_{Q_x}(o=1|p=1)\ge\frac{2}{3}$,
\item
if $x\notin L$, then 
$P_{Q_x}(o=1|p=1)\le\frac{1}{3}$.
\end{itemize}
Note that the error bound $(2/3,1/3)$
can be amplified to $(1-2^{-r(n)},2^{-r(n)})$ for any polynomial $r$
by using the standard amplification technique.
\end{definition}

Second, a variant of postBQP, which is called
${\rm postBQP}_{\rm FP}$, was defined in 
Ref.~\cite{MN}, and shown to be in AWPP.

\begin{definition}
A language $L$ is in ${\rm postBQP}_{\rm FP}$ if and only if
there exist a polynomial $s$, an FP function (i.e., polynomial-time
computable function) $f$, 
and a uniform family
$\{Q_x\}_x$ of polynomial-size
quantum circuits, where $x$ is an instance with $|x|=n$,
such that
\begin{eqnarray*}
P_{Q_x}(p=1)=\frac{f(x)}{2^{s(n)}}
\end{eqnarray*}
and
\begin{itemize}
\item
if $x\in L$, then 
$P_{Q_x}(o=1|p=1)\ge\frac{2}{3}$,
\item
if $x\notin L$, then 
$P_{Q_x}(o=1|p=1)\le\frac{1}{3}$.
\end{itemize}
Note that the error bound $(2/3,1/3)$
can be amplified to $(1-2^{-r(n)},2^{-r(n)})$ for any polynomial $r$
by using the standard amplification technique.
Furthermore, it was shown in Ref.~\cite{MN} that $f$ can be actually taken to
be 1 without changing the power of the class.
\end{definition}

Finally, the class SBQP was defined by Kuperberg~\cite{Kuperberg}. It is
a quantum version of SBP~\cite{BGM}, 
and equal to the classical class A$_0$PP~\cite{Vya}.
\begin{definition}
A language $L$ is in {\rm SBQP} if and only if there exist
a polynomial $r$ and a uniform family $\{Q_x\}_x$ 
of polynomial-size quantum circuits, where $x$ is an instance
with $|x|=n$, 
such that
\begin{itemize}
\item
If $x\in L$, then $Q_x$ accepts with probability at least $2^{-r(n)}$.
\item
If $x\notin L$, then $Q_x$ accepts with probability at most $2^{-r(n)-1}$.
\end{itemize}
Note that the error bound $(2^{-r(n)},2^{-r(n)-1})$
can be replaced with $(a2^{-r(n)},b2^{-r(n)})$ for any 
$0\le b< a\le1$ such that $a-b\ge1/poly(n)$.
\end{definition}

\section{New classes}
\label{new_classes}

In this section, we define new classes that we study.

\begin{definition}
A language $L$ is in 
${\rm QMA}_{\rm postBQP}$ if and only if there exist polynomials
$w$, $m$, and $s$, and a uniform family $\{Q_x\}_x$ 
of polynomial-size
quantum circuits,
where $x$ is an instance with $|x|=n$, 
$Q_x$ takes as input a $w(n)$-qubit quantum state (so called the witness), 
and $m(n)$ ancilla qubits initialized in $|0\rangle$,
such that
\begin{eqnarray*}
P_{Q_x(\xi)}(p=1)\ge\frac{1}{2^{s(n)}}
\end{eqnarray*}
for any $w(n)$-qubit state $\xi$,
and
\begin{itemize}
\item 
if $x\in L$, 
then there exists a $w(n)$-qubit quantum
state $\psi$ such that 
\begin{eqnarray*}
P_{Q_x(\psi)}(o=1|p=1)\ge\frac{2}{3}.
\end{eqnarray*}
\item
if $x\notin L$, then for any $w(n)$-qubit quantum state 
$\xi$,
\begin{eqnarray*}
P_{Q_x(\xi)}(o=1|p=1)\le\frac{1}{3}.
\end{eqnarray*}
\end{itemize}
Note that if we are allowed to increase the
witness length $w$, we can amplify the error bound
$(2/3,1/3)$ to $(1-2^{-r(n)},2^{-r(n)})$ for any polynomial $r$ by using 
the standard amplification technique.
(It is open whether the Marriott-Watrous type 
amplification~\cite{MW}
is possible for this class.)
\end{definition}

Like QMA, the yes witness state can be restricted to
be pure:

\begin{lemma}
In the definition of ${\rm QMA}_{\rm postBQP}$, the yes witness $\psi$
can be a pure state without changing the power of the class.
\end{lemma}

\begin{proof}
Let us assume that
\begin{eqnarray*}
P_{Q_x(\psi)}(o=1|p=1)\ge\frac{2}{3}
\end{eqnarray*}
for a state $\psi$. Let us diagonalize $\psi$ as 
$\psi=\sum_i\alpha_i|\psi_i\rangle\langle \psi_i|$ with
eigenvalues $\{\alpha_i\}_i$ and eigenvectors $\{|\psi_i\rangle\}_i$.
Let us assume that
\begin{eqnarray*}
P_{Q_x(|\psi_i\rangle)}(o=1|p=1)<\frac{2}{3}
\end{eqnarray*}
for all $i$. Then,
\begin{eqnarray*}
P_{Q_x(\psi)}(o=1|p=1)&=&
\frac{P_{Q_x(\psi)}(o=1,p=1)}{P_{Q_x(\psi)}(p=1)}\\
&=&\sum_i\alpha_i\frac{P_{Q_x(|\psi_i\rangle)}(o=1,p=1)}{P_{Q_x(\psi)}(p=1)}\\
&=&\sum_i\alpha_i\frac{P_{Q_x(|\psi_i\rangle)}(o=1,p=1)}
{P_{Q_x(|\psi_i\rangle)}(p=1)}
\frac{P_{Q_x(|\psi_i\rangle)}(p=1)}{P_{Q_x(\psi)}(p=1)}\\
&<&\frac{2}{3}\sum_i\alpha_i
\frac{P_{Q_x(|\psi_i\rangle)}(p=1)}{P_{Q_x(\psi)}(p=1)}\\
&=&\frac{2}{3}
\frac{P_{Q_x(\psi)}(p=1)}{P_{Q_x(\psi)}(p=1)}\\
&=&\frac{2}{3},
\end{eqnarray*}
which contradicts to the assumption. Therefore,
\begin{eqnarray*}
P_{Q_x(|\psi_i\rangle)}(o=1|p=1)\ge\frac{2}{3}
\end{eqnarray*}
for at least one pure state $|\psi_i\rangle$.
\end{proof}

\begin{definition}
The class ${\rm QCMA}_{\rm postBQP}$ is defined similarly to
${\rm QMA}_{\rm postBQP}$ except that the witness is not
a $w(n)$-qubit state but a classical $w(n)$-bit string (or,
equivalently, a $w(n)$-qubit state in the computational basis).
\end{definition}

\begin{definition}
A language $L$ is in ${\rm QMA}_{\rm postBQP}^*$
if and only if it is in ${\rm QMA}_{\rm postBQP}$
and 
\begin{eqnarray*}
P_{Q_x(\xi)}(p=1)=
P_{Q_x(\rho)}(p=1)
\end{eqnarray*}
for any $w(n)$-qubit states $\xi$ and $\rho$.
\end{definition}

\begin{definition}
The class ${\rm QCMA}_{\rm postBQP}^*$ is defined similarly to
${\rm QMA}_{\rm postBQP}^*$ except that the witness is not
a $w(n)$-qubit state but a classical $w(n)$-bit string (or,
equivalently, a $w(n)$-qubit state in the computational basis).
\end{definition}

\begin{definition}
A language $L$ is in 
${\rm QMA}_{\rm SBQP}$ if and only if there exist polynomials
$w$, $m$, and $r$, and a uniform family $\{Q_x\}_x$ 
of polynomial-size quantum circuits,
where $x$ is an instance with $|x|=n$, 
$Q_x$ takes as input a $w(n)$-qubit quantum state (so called the witness), 
and $m(n)$ ancilla qubits initialized in $|0\rangle$,
such that
\begin{itemize}
\item
if $x\in L$, 
then there exists a $w(n)$-qubit quantum
state $\psi$ such that 
\begin{eqnarray*}
P_{Q_x(\psi)}(o=1)\ge2^{-r(n)}.
\end{eqnarray*}
\item
if $x\notin L$, then for any $w(n)$-qubit quantum state $\xi$,
\begin{eqnarray*}
P_{Q_x(\xi)}(o=1)\le2^{-r(n)-1}.
\end{eqnarray*}
\end{itemize}
Note that, without loss of generality, we can assume that the yes witness
$\psi$ is a pure state.
Furthermore, note that the error bound $(2^{-r(n)},2^{-r(n)-1})$
can be amplified to $(2^{-r(n)k},2^{-r(n)k-k})$
for any integer $k\ge1$ without changing the witness size $w$
by using a similar technique of Ref.~\cite{MW}.
(In Ref.~\cite{MW}, we accept if $\sum_{i=1}^Nz_i\ge N\frac{a+b}{2}$,
but now we accept if all $z_i=1$.)
\end{definition}

\begin{definition}
The class ${\rm QCMA}_{\rm SBQP}$ is defined similarly to
${\rm QMA}_{\rm SBQP}$ except that the witness is not
a $w(n)$-qubit state but a classical $w(n)$-bit string (or,
equivalently, a $w(n)$-qubit state in the computational basis).
\end{definition}

\begin{definition}
A language $L$ is in 
${\rm QMA}_{\rm postBQP_{\rm FP}}$ if and only if 
it is in
${\rm QMA}_{\rm postBQP}$ and
\begin{eqnarray*}
P_{Q_x(\xi)}(p=1)=\frac{1}{2^{s(n)}}
\end{eqnarray*}
for any $w(n)$-qubit state $\xi$.
Here, $s$ is the polynomial determined from the definition
of ${\rm QMA}_{\rm postBQP}$.
\end{definition}

It is obvious that 
${\rm QMA}\subseteq{\rm QMA}_{{\rm postBQP}_{\rm FP}}$.
Showing the equality,
${\rm QMA}={\rm QMA}_{{\rm postBQP}_{\rm FP}}$,
seems to be difficult, since 
${\rm WPP}\subseteq{{\rm postBQP}_{\rm FP}}$~\cite{MN},
and therefore the equality leads to ${\rm WPP}\subseteq{\rm QMA}$.
The class ${\rm WPP}$ contains ${\rm SPP}$,
and ${\rm SPP}$ contains the graph (non)isomorphism.
It is an open problem whether the graph non-isomorphism 
is in QMA~\cite{Watrous}.

\begin{definition}
The class ${\rm QCMA}_{{\rm postBQP}_{\rm FP}}$ is defined similarly to
${\rm QMA}_{{\rm postBQP}_{\rm FP}}$ except that the witness is not
a $w(n)$-qubit state but a classical $w(n)$-bit string (or,
equivalently, a $w(n)$-qubit state in the computational basis).
\end{definition}

Obviously,
${\rm postBQP}_{\rm FP}\subseteq{\rm QCMA}_{{\rm postBQP}_{\rm FP}}$.
The equality,
${\rm postBQP}_{\rm FP}={\rm QCMA}_{{\rm postBQP}_{\rm FP}}$, seems to
be unlikely, since it leads to
\begin{eqnarray*}
{\rm NP}\subseteq{\rm QCMA}\subseteq{\rm QCMA}_{{\rm postBQP}_{\rm FP}}
={\rm postBQP}_{\rm FP}\subseteq{\rm AWPP},
\end{eqnarray*}
but it is known that there exists an oracle 
$A$ such that ${\rm P}^A={\rm AWPP}^A$ and the polynomial hierarchy
is infinite~\cite{FFKL}.

Furthermore, it is obvious that 
${\rm QCMA}\subseteq{\rm QCMA}_{{\rm postBQP}_{\rm FP}}$.
Again, showing the equality,
${\rm QCMA}={\rm QCMA}_{{\rm postBQP}_{\rm FP}}$,
seems to be difficult, since it leads to
${\rm WPP}\subseteq{\rm QCMA}$.

\section{Results}
In this section, we give the results of this paper.
We show several relations between our new complexity classes
and known complexity classes.

We first study 
${\rm QMA}_{\rm postBQP}$ and
${\rm QCMA}_{\rm postBQP}$.

\begin{theorem}
${\rm QMA}_{\rm postBQP}\subseteq 
{\rm QMA}(\frac{1}{2}+2^{-r},\frac{1}{2}-2^{-r})$
for a polynomial $r$.
\end{theorem}

\begin{proof}

Let us assume that a language $L$ is in ${\rm QMA}_{\rm postBQP}$,
and let $Q_x$ be Arthur's circuit.
Let us consider the following circuit $R_x$:
\begin{itemize}
\item[1.]
It simulates $Q_x$ on input (witness) $\xi$.
\item[2.]
If $Q_x$ outputs $o=0$ and $p=0$, then 
$R_x$ outputs $o=1$ with probability 1/2 and $o=0$ with probability 1/2.
\item[3.]
If $Q_x$ outputs $o=0$ and $p=1$, $R_x$ outputs $o=0$.
\item[4.]
If $Q_x$ outputs $o=1$ and $p=0$, $R_x$ outputs $o=1$ with probability 1/2
and $o=0$ with probability 1/2.
\item[5.]
If $Q_x$ outputs $o=1$ and $p=1$, $R_x$ outputs $o=1$.
\end{itemize}
Then,
\begin{eqnarray*}
P_{R_x(\xi)}(o=1)=
P_{Q_x(\xi)}(o=1,p=1)
+\frac{1}{2}
P_{Q_x(\xi)}(p=0).
\end{eqnarray*}

If $x\in L$, then by the assumption of $L\in {\rm QMA}_{\rm postBQP}$,
there exists a $w(n)$-qubit state $\psi$ such that
\begin{eqnarray*}
P_{Q_x(\psi)}(o=1|p=1)-P_{Q_x(\psi)}(o=0|p=1)\ge\frac{1}{2}.
\end{eqnarray*}
If we multiply both sides by $P_{Q_x(\psi)}(p=1)$, 
we obtain
\begin{eqnarray*}
P_{Q_x(\psi)}(o=1,p=1)- P_{Q_x(\psi)}(o=0,p=1)\ge
\frac{1}{2}P_{Q_x(\psi)}(p=1).
\end{eqnarray*}
By the assumption, $P_{Q_x(\psi)}(p=1)\ge 2^{-s}$
for some polynomial $s$.
Therefore,
\begin{eqnarray*}
P_{Q_x(\psi)}(o=1,p=1)- P_{Q_x(\psi)}(o=0,p=1)\ge
\frac{1}{2}2^{-s}.
\end{eqnarray*}
Hence
\begin{eqnarray*}
P_{Q_x(\psi)}(o=1,p=1)-(1-P_{Q_x(\psi)}(o=1,p=1)-P_{Q_x(\psi)}(p=0))
\ge\frac{1}{2}2^{-s},
\end{eqnarray*}
which means
\begin{eqnarray*}
2P_{Q_x(\psi)}(o=1,p=1)+P_{Q_x(\psi)}(p=0)\ge1+\frac{1}{2}2^{-s}.
\end{eqnarray*}
Therefore, we obtain
\begin{eqnarray*}
P_{R_x(\psi)}(o=1)\ge \frac{1}{2}+\frac{1}{4}2^{-s}.
\end{eqnarray*}

If $x\notin L$, on the other hand,
for any $w(n)$-qubit state $\xi$ 
\begin{eqnarray*}
P_{Q_x(\xi)}(o=0|p=1)- P_{Q_x(\xi)}(o=1|p=1)\ge\frac{1}{2}.
\end{eqnarray*}
In a similar way, this means
\begin{eqnarray*}
P_{R_x(\xi)}(o=1)\le \frac{1}{2}-\frac{1}{4}2^{-s}.
\end{eqnarray*}
Therefore, $L$ is in ${\rm QMA}(\frac{1}{2}+2^{-s-2},\frac{1}{2}-2^{-s-2})$.
\end{proof}

According to Refs.~\cite{FL1,FL2,FKLMN},
${\rm QMA}(\frac{1}{2}+2^{-r},\frac{1}{2}-2^{-r})\subseteq
{\rm PSPACE}
$
for any polynomial $r$.
Therefore, the above theorem means
\begin{eqnarray*}
{\rm QMA}_{\rm postBQP}\subseteq{\rm PSPACE}.
\end{eqnarray*}

\begin{theorem}
\label{preciseQMA in postQMA}
$\mathrm{QMA}(\frac{1}{2}+2^{-r},\frac{1}{2}-2^{-r})\subseteq
\mathrm{QMA}_{\rm postBQP}$
for any polynomial $r$.
\end{theorem}

\begin{proof}
Let $L$ be a language in 
$\mathrm{QMA}(\frac{1}{2}+2^{-r},\frac{1}{2}-2^{-r})$ for a polynomial $r$. 
Let $V_x$ be Arthur's circuit verifying $L$. 
Without loss of generality, we can assume that the maximum acceptance probability 
of $V_x$ (over quantum witnesses) is at most $1-2^{-r}$ 
(by modifying the original system so that it can be accepted and rejected 
automatically with an exponentially small probability). Let $|\varphi_x\rangle$ be a quantum witness that achieves 
the maximum acceptance probability of $V_x$. Then, we have
\[
V_x(|\varphi_x\rangle\otimes|0\rangle^{\otimes m})=
\sqrt{p_x}|0\rangle\otimes|\phi_{x,0}\rangle+
\sqrt{1-p_x}|1\rangle\otimes|\phi_{x,1}\rangle
\]
for certain $(w+m-1)$-qubit states 
$|\phi_{x,0}\rangle$ and $|\phi_{x,1}\rangle$, 
where $p_x$ is the maximum acceptance probability of $V_x$.
Now by the DISTILLATION PROCEDURE of Ref.~\cite{KLGN15SIAM} 
(see Subsection 6.1.1 in \cite{KLGN15SIAM}) we can obtain a single-qubit state
\[
|\psi\rangle = \frac{1}{\sqrt{p_x^2+(1-p_x^2)}}(p_x|0\rangle+(1-p_x)|1\rangle)
\]
using postselection with probability $p_x^2+(1-p_x)^2$ ($=2p_x^2-2p_x+1$).

The rest of the proof is similar to that of 
$\mathrm{PP}\subseteq\mathrm{postBQP}$~\cite{Aaronson}. 
Let $H$ be the Hadamard gate.
For some positive real numbers $\alpha,\beta$ to be specified later,
prepare $\alpha|0\rangle|\psi\rangle+\beta|1\rangle H|\psi\rangle$ where
\[
H|\psi\rangle= \frac{1}{\sqrt{p_x^2+(1-p_x^2)} }
\Big(\frac{1}{\sqrt{2}}|0\rangle+\frac{2p_x-1}{\sqrt{2}}|1\rangle\Big).
\]
Then postselect on the second qubit being $|1\rangle$. This gives the reduced state
\[
|\varphi_{\beta/\alpha}\rangle = \frac{\alpha(1-p_x)|0\rangle+\beta\sqrt{1/2}(2p_x-1)|1\rangle}{\sqrt{\alpha^2(1-p_x)^2+\frac{\beta^2}{2}(2p_x-1)^2 }}
\]
in the first qubit.

Suppose $x\in L$. Then, $p_x\geq 1/2+1/2^r$ (and $p_x\leq 1-1/2^r$ by the assumption). $1-p_x>0$ and $\sqrt{1/2}(2p_x-1)>0$ and hence (the pair of the two real coefficients of) $|\varphi_{\beta/\alpha}\rangle$ lies in the first quadrant. Then we claim there exists an integer $i\in [-r,r]$ such that if we set $\beta/\alpha=2^i$, then $|\varphi_{2^i}\rangle$ is close 
to $|+\rangle=\frac{1}{\sqrt{2}}(|0\rangle+|1\rangle)$: $|\langle +|\varphi_{2^i}\rangle|\geq (1+\sqrt{2})/\sqrt{6}>0.985$. 
In fact, since the ratio $\frac{\sqrt{1/2}(2p_x-1)}{1-p_x}$ lies between $1/2^r=2^{-r}$ 
and $2^r$, there must be an integer $i\in [-r,r-1]$ such that $|\varphi_{2^i}\rangle$ 
and $|\varphi_{2^{i+1}}\rangle$ fall on the opposite sides of $|+\rangle$ in the first quadrant.
Thus the worst case is that $\langle +|\varphi_{2^i}\rangle=\langle +|\varphi_{2^{i+1}}\rangle$,
which occurs when $|\varphi_{2^i}\rangle=\sqrt{2/3}|0\rangle+\sqrt{1/3}|1\rangle$ 
and $|\varphi_{2^{i+1}}\rangle=\sqrt{1/3}|0\rangle+\sqrt{2/3}|1\rangle$.

On the contrary, suppose $x\notin L$. Then, $p_x\leq 1/2-1/2^r$.
Thus, $1-p_x>0$ and $\sqrt{1/2}(2p_x-1)<0$ and hence $|\varphi_{\beta/\alpha}\rangle$ lies in the fourth quadrant. Then $|\varphi_{2^i}\rangle$ never lies in the first or the third quadrants and therefore $|\langle +|\varphi_{2^i}\rangle|\leq 1/\sqrt{2}<0.708$.
Moreover, if Merlin sends a state which does not correspond to the maximum acceptance probability $p_x$ of $V_x$, by the DISTILLATION PROCEDURE 
we obtain a mixture of states in the form of 
\[
|\psi'\rangle = \frac{1}{\sqrt{q^2+(1-q^2)}}(q|0\rangle+(1-q)|1\rangle)
\] 
where $q\leq p_x$ with postselection (as seen from the analysis of Subsection 6.1.1 in \cite{KLGN15SIAM}). Thus also in this case we can obtain the same conclusion of $|\langle +|\varphi_{2^i}\rangle|\leq 1/\sqrt{2}<0.708$.

It follows that, by repeating the whole algorithm $r(2r+1)$ times with $r$ invocations for each integer $i\in [-r,r]$, we can learn whether $x\in L$ or $x\notin L$ 
with exponentially small probability of error (by the standard analysis 
of the error reduction of QMA proof systems).
\end{proof}

Fefferman and Li~\cite{FL1,FL2} 
showed that 
\begin{eqnarray*}
\bigcup_{r:polynomial}\mathrm{QMA}\Big(\frac{1}{2}+2^{-r},\frac{1}{2}-2^{-r}\Big)
=\mathrm{PSPACE}. 
\end{eqnarray*}
Therefore, the above theorem
means 
\begin{eqnarray*}
\mathrm{PSPACE}\subseteq\mathrm{QMA}_{\rm postBQP}.
\end{eqnarray*}

Combining the two theorems, we have our first main result:
\begin{theorem}
${\rm QMA}_{\rm postBQP}={\rm PSPACE}$.
\end{theorem}

For characterizing ${\rm QCMA}_{\rm postBQP}$, 
let us recall $\exists$ operator as follows.

\begin{definition}
Let ${\rm C}$ be a class.
A language $L$ is in $\exists {\rm C}$ if and only if
there exist a language $L'\in {\rm C}$ and a polynomial $q$
such that
\begin{itemize}
\item[1.]
If $x\in L$, then there exists a string $y$ of length $q(|x|)$
such that $\langle x, y\rangle\in L'$.
\item[2.]
If $x\notin L$, then for any string $y$ of length $q(|x|)$,
$\langle x, y\rangle\notin L'$.
\end{itemize}
\end{definition}

Then, we can observe:

\begin{proposition}
$\exists {\rm PP}={\rm QCMA}_{\rm postBQP}$.
\end{proposition}

\begin{proof}
It is obvious by the definition of $\exists{\rm PP}$
and ${\rm PP}={\rm postBQP}$.
\end{proof}

Now let us move on to 
${\rm QMA}_{\rm postBQP}^*$
and
${\rm QCMA}_{\rm postBQP}^*$.
These classes are shown to coincide with PP.

\begin{theorem}
${\rm QCMA}_{\rm postBQP}^*={\rm QMA}_{\rm postBQP}^*={\rm PP}$.
\end{theorem}

\begin{proof}
First,
${\rm PP}={\rm postBQP}\subseteq{\rm QCMA}_{\rm postBQP}^*$
is trivial, since Arthur has only to ignore the witness.

Next, we show
${\rm QMA}_{\rm postBQP}^*\subseteq{\rm PP}$.
Let us assume that a language $L$ is in ${\rm QMA}_{\rm postBQP}^*$,
and let $Q_x$ be Arthur's quantum circuit that recognizes $L$.
If $x\in L$, then there exists a $w(n)$-qubit pure state $\psi$ such that
\begin{eqnarray*}
P_{Q_x(\psi)}(o=1,p=1)&\ge& \frac{2}{3}P_{Q_x(\psi)}(p=1).
\end{eqnarray*}
By using the technique of Ref.~\cite{MW}
(more precisely, the AND-Repetition procedure in Ref.~\cite{FKLMN}), 
we can construct
for any $k$ a circuit $R_x$ such that
\begin{eqnarray*}
P_{R_x(\psi)}(o=1,p=1)&\ge&\Big(\frac{2}{3}\Big)^kP_{Q_x(\psi)}(p=1)^k.
\end{eqnarray*}
Therefore,
\begin{eqnarray*}
P_{R_x(I^w/2^w)}(o=1,p=1)&\ge&2^{-w}P_{R_x(\psi)}(o=1,p=1)\nonumber\\
&\ge&2^{-w}\Big(\frac{2}{3}\Big)^kP_{Q_x(\psi)}(p=1)^k\nonumber\\
&=&2^{-w}\Big(\frac{2}{3}\Big)^kP_{Q_x(I^w/2^w)}(p=1)^k\nonumber\\
&=&\frac{4}{3}\times\frac{1}{3^{w+1}}P_{Q_x(I^w/2^w)}(p=1)^{w+2}\label{eq1},
\end{eqnarray*}
where we have taken $k=w+2$.

On the other hand, if $x\notin L$, 
\begin{eqnarray*}
P_{Q_x(\xi)}(o=1,p=1)&\le& \frac{1}{3}P_{Q_x(\xi)}(p=1)
\end{eqnarray*}
for any state $\xi$,
and therefore
\begin{eqnarray*}
P_{R_x(I^w/2^w)}(o=1,p=1)&\le&
\Big(\frac{1}{3}\Big)^kP_{Q_x(I^w/2^w)}(p=1)^k\nonumber\\
&=&\frac{1}{3}\times\frac{1}{3^{w+1}}P_{Q_x(I^w/2^w)}(p=1)^{w+2}\nonumber\\
&\le&\frac{3}{4}\times\frac{1}{3^{w+1}}P_{Q_x(I^w/2^w)}(p=1)^{w+2}\label{eq2}.
\end{eqnarray*}
Therefore, due to the definition of postBQP by Kuperberg~\cite{Kuperberg},
$L$ is in ${\rm postBQP}={\rm PP}$.

\end{proof}

It is known that $\exists {\rm PP}={\rm NP}^{\rm PP}$~\cite{Toran}.
Therefore,
from Toda's theorem \cite{Toda}, 
${\rm QCMA}_{\rm postBQP}=\exists {\rm PP}={\rm NP}^{\rm PP}$
contains the polynomial hierarchy.
On the other hand,
${\rm Q(C)MA}_{\rm postBQP}^*={\rm PP}$.
Therefore, it seems that
${\rm Q(C)MA}_{\rm postBQP}^*
\neq {\rm QCMA}_{\rm postBQP}$.

We next consider
the Merlinizations of SBQP,
${\rm QMA}_{\rm SBQP}$
and
${\rm QCMA}_{\rm SBQP}$.

\begin{theorem}
\label{result1}
${\rm QCMA}_{\rm SBQP}={\rm QMA}_{\rm SBQP}={\rm SBQP}$.
\end{theorem}

\begin{proof}
${\rm SBQP}\subseteq{\rm QCMA}_{\rm SBQP}$ is obvious.
Let us show
${\rm QMA}_{\rm SBQP}\subseteq{\rm SBQP}$.
Let us assume that a language $L$ is in ${\rm QMA}_{\rm SBQP}$, and let
$Q_x$ be Arthur's circuit that recognizes $L$.
Let $w$, $m$ and $r$ be the polynomials determined from the definition
of ${\rm QMA}_{\rm SBQP}$.
We construct an SBQP algorithm that recognizes $L$.
In our SBQP algorithm, we run $Q_x$ on 
$\frac{I^{\otimes w}}{2^w}\otimes |0\rangle\langle0|^{\otimes m}$.
If $x\in L$, for any $k$,
\begin{eqnarray*}
P_{Q_x(I^w/2^w)}(o=1)&\ge& 
2^{-w}P_{Q_x(\psi)}(o=1)\\
&\ge&
2^{-w-rk}.
\end{eqnarray*}
If $x\notin L$,
\begin{eqnarray*}
P_{Q_x(I^w/2^w)}(o=1)\le 2^{-rk-k}.
\end{eqnarray*}
Therefore,
if we take $k=w+1$, we obtain
\begin{eqnarray*}
P_{Q_x(I^w/2^w)}(o=1)
\left\{
\begin{array}{ll}
\ge 2^{-w-r(w+1)}&(x\in L)\\
\le 2^{-1}2^{-w-r(w+1)}&(x\notin L),
\end{array}
\right.
\end{eqnarray*}
which means that $L$ is in SBQP.
\end{proof}

Finally, let us consider the Merlinized version of
${\rm postBQP}_{\rm FP}$.

\begin{theorem}
${\rm QMA}_{{\rm postBQP}_{\rm FP}}\subseteq{\rm SBQP}$.
\end{theorem}

Note that previous theorem shows ${\rm QMA}_{\rm SBQP}={\rm SBQP}$
and it is known that ${\rm postBQP}_{\rm FP}\subseteq{\rm SBQP}$.
Therefore, one might think that the relation
${\rm QMA}_{\rm postBQP_{\rm FP}}\subseteq{\rm SBQP}$
is trivially derived
from these two facts. However, an inclusion relation for
the verifier does not necessarily mean that for the language class,
and therefore we provide a proof below.

\begin{proof}
We assume that a language $L$ is in ${\rm QMA}_{{\rm postBQP}_{\rm FP}}$.
Let $Q_x$ be Arthur's circuit that recognizes $L$.
Then, if $x\in L$, there exist a polynomial $s$
and a $w(n)$-qubit pure state $\psi$ such that
\begin{eqnarray*}
P_{Q_x(\psi)}(o=1,p=1)\ge\frac{2}{3}2^{-s},
\end{eqnarray*}
and if $x\notin L$, 
\begin{eqnarray*}
P_{Q_x(\xi)}(o=1,p=1)\le\frac{1}{3}2^{-s},
\end{eqnarray*}
for any $w(n)$-qubit state $\xi$.

Now we construct an SBQP algorithm that recognizes $L$. For the goal,
we consider the new circuit $R_x$ that can amplify the error bound
without changing the witness size $w$ by using the Marriott-Watrous
technique~\cite{MW} 
(or the AND-Repetition procedure of Ref.~\cite{FKLMN}): for any integer $k$,
if $x\in L$, there exists a $w(n)$-qubit state $\psi$ such that
\begin{eqnarray*}
P_{R_x(\psi)}(o=1,p=1)\ge\Big(\frac{2}{3}\Big)^k2^{-ks},
\end{eqnarray*}
and if $x\notin L$, 
\begin{eqnarray*}
P_{R_x(\xi)}(o=1,p=1)\le\Big(\frac{1}{3}\Big)^k2^{-ks},
\end{eqnarray*}
for any $w(n)$-qubit state $\xi$.

If we run $R_x$ on 
$\frac{I^{\otimes w}}{2^w}\otimes|0\rangle\langle0|^{\otimes m}$,
where $m$ is the number of ancilla qubits,
we obtain if $x\in L$,
\begin{eqnarray*}
P_{R_x(I^w/2^w)}(o=1,p=1)&\ge& 
2^{-w}P_{R_x(\psi)}(o=1,p=1)\\
&\ge&2^{-w}\Big(\frac{2}{3}\Big)^k2^{-ks}\\
&=&2^{-w}\Big(\frac{2}{3}\Big)^{w+1}2^{-(w+1)s}\\
&=&2\times{3}^{-w-1}2^{-(w+1)s},
\end{eqnarray*}
and if $x\notin L$,
\begin{eqnarray*}
P_{R_x(I^w/2^w)}(o=1,p=1)&\le& 
\Big(\frac{1}{3}\Big)^k2^{-ks}\\
&=&3^{-w-1}2^{-(w+1)s},
\end{eqnarray*}
where we have taken $k=w+1$. 
Hence, $L$ is in SBQP.
\end{proof}

\section{Note added}
After completing the draft,
we have noticed the paper by Usher, Hoban, and 
Browne~\cite{UsherHobanBrowne}.
The class, postQMA, defined by them is the same as
our class, QMA$_{\rm postBQP}$. Although they remain
the upperbound and lowerbound of postQMA as open,
we here show that it is equal to PSPACE.
The class, postQMA$^*$, defined by them is also the same as
our class, QMA$_{\rm postBQP}^*$.
They show that postQMA$^*$ is in PP by using GapP functions,
while we here show that QMA$_{\rm postBQP}^*$ is in PP by using
another definition of PP by Kuperberg~\cite{Kuperberg}.

\acknowledgements
We thank Hirotada Kobayashi for helpful discussion.
TM is supported by JST ACT-I, 
the JSPS Grant-in-Aid for Young Scientists (B) No.26730003 and
No.17K12637,
and the MEXT JSPS Grant-in-Aid for Scientific Research on
Innovative Areas No.15H00850.
HN is supported by the JSPS Grant-in-Aid for Scientific Research (A) 
Nos.26247016, 16H01705 and (C) No.16K00015,
and the MEXT JSPS Grant-in-Aid for Scientific Research on
Innovative Areas No.24106009.

\end{document}